\documentclass[11pt]{article}%
\usepackage{amssymb}
\usepackage{amsfonts}
\usepackage{amsmath}
\usepackage{amssymb}
\usepackage{color}
\usepackage{graphicx}
\usepackage{caption}%
\usepackage{natbib}
\usepackage{hyperref}
\usepackage{float}
\usepackage{subcaption}

\usepackage{tabularx}
\usepackage{algorithm}
\usepackage{algpseudocode}
\providecommand{\U}[1]{\protect\rule{.1in}{.1in}}
\newtheorem{theorem}{Theorem}

\newtheorem{definition}[theorem]{Definition}
\newenvironment{proof}[1][Proof]{\noindent \textbf{#1.} }{\  \rule{0.5em}{0.5em}}
\begin{document}
	
	\title{\textbf{Robust Rao-type tests for step-stress accelerated life-tests under  interval-monitoring and Weibull lifetime distributions}}
	\author{Narayanaswamy Balakrishnan,  María Jaenada and Leandro Pardo}
	\date{ }
	\maketitle

\begin{abstract}
	Many products in engineering are highly reliable with large mean lifetimes to failure. Performing lifetests under normal operations conditions would thus require long experimentation times and high experimentation costs. Alternatively, accelerated lifetests shorten the experimentation time by running the tests at higher than normal stress conditions, thus inducing more failures. 
	Additionally, a log-linear regression model can be used to relate the lifetime distribution of the product to the level of stress it experiences. After estimating the parameters of this relationship, results can be extrapolated to normal operating conditions.
	On the other hand, censored data is common in reliability analysis. Interval-censored data arise when continuous inspection is difficult or infeasible due to technical or budgetary constraints. 
	In this paper, we develop robust restricted estimators based on the density power divergence for step-stress accelerated life-tests under Weibull distributions with interval-censored data. We present theoretical asymptotic properties of the estimators and develop robust Rao-type test statistics based on the proposed robust estimators for testing composite null hypothesis on the model parameters.
	
	 
\end{abstract}
\section{Introduction}

Inference under censoring is a key concern  in reliability and survival analysis because of its high occurrence in lifetime data.
Censoring occurs when incomplete information is available about the survival time of some devices because it has not been possible to observe the exact time of failure. Instead, either the interval in which it occurred is known (interval censoring), a lower bound of the failure time is known (Censored Type I data or right-censored data) or an upper bound of the failure time is known. 
Because lifetime information is incomplete, inference under censored data complicates the estimation of the lifetime distribution
and it can result in biased estimates of product lifetime characteristics.  Further, different types of censoring require different statistical methods to address them properly, and careful consideration and appropriate handling of censoring are essential to ensure accurate and reliable results. Here, we deal with interval-censored data, which are commonly found in reliability experiment with interval-monitoring  where failure counts are recorded at certain pre-fixed monitoring times, or experiments where functional check requires a manual test.

In engineering reliability analysis we are often interested in inferring the reliability and associated lifetime characteristics of a certain product. From the parametric approach,  a parametric distribution family is assumed for modelling the time until failure. Two of the most common lifetime distributions adaopted in the literature are the exponential and the Weibull distribution, as they are the simplest distributions  characterizing survival/reliability data. The exponential distribution is memoryless, and thus it hazard function is constant with respect to time, which makes reliabililty analysis very simple. The Weibull distribution is a flexible distribution that can model a wide range of failure patterns and includes the exponential distribution as a particular case. Of course, both parametric families depend on model parameters that need to be estimated through observed data. The maximum likelihood estimator (MLE) is generally considered because of it good properties in terms of unbiasness, consistency and asymptotic efficiency. However, the MLE lacks of robustness and so data contamination may heavily affect the estimation.
Recent works on survival analysis and reliability have shown the advantage in terms of robustness, with a small loss of efficiency, of using divergence-based estimators. 
  \cite{BalaJae2022a} developed robust minimum density power divergence (DPD) estimators (MDPDEs) for non-destructive one-shot device data under step-stress model with exponential lifetimes, and later  \cite{BalaJae2022b, BalaJae2023} extended the theory to Weibull and gamma lifetime distributions.
Further, \cite{BalaJae2023a} developed restricted MDPDE (RMDPDE) for the step-stress model with exponential lifetimes, and derived robust Rao-type tests based the restricted estimators.

Besides, as many nowadays products are highly reliable,  it would be highly unlikely to observe many failures on tests under normal operating conditions within a short period of time.  Alternatively, an accelerated life test (ALT) can be used to shorten the reliability of the devices under test by increasing the level of stress to which devices are subjected. Then,  after analyzing the observed data  under high-stress conditions, the inferential results can be extrapolated from high stress conditions to normal operating conditions.
There are different types of ALT plans resulting in different statistical models; in constant-stress ALTs each unit under test is subjected to a (single) stress level, step-stress ALTs increase the stress level to all units under test at certain fixed times and progressive-stress ALTs continuously increase the stress level during the test. 
Continuous inspection of the exact failure times of the devices under test may not be possible in practice due to technical limitations of the product or either may be too expensive. To reduce experimentation cost, failure counts can be collected at some inspection times during the test, resulting in interval-censored data.
This is the case of solar lighting devices, where it may be preferable to examine the life state of the devices under test only at certain inspection times, rather than continuously during the experimentation.
 \cite{Han2014} conducted a simple step-stress test  under a time constraint to assess the reliability of a solar lighting device, and  \cite{Han2019, Han2020} analysed solar light data under interval-monitoring. In those studies, exponential and Weibull lifetime distributions were considered for the failure times.
We focus here on the step-stress model for interval-monitored lifetimes from Weibull distributions.

Since in step stress experiments the stress level is increased, a statistical model relating   the lifetime distribution at a current stress level to the preceding levels is required.
 \cite{Nelson1980} proposed the cumulative exposure model (CE), that states the residual life of a product does only depend on the current cumulative fraction failed and the cumulative distribution function of the current stress.
It has no memory on how wore and tear was accumulated and so the successive failure behaviour is independent of the previous stress levels, except for the cumulative exposure generated until the time of stress change.
Thus, for any two assumed lifetime distributions $F_1(t)$ and $F_2(t)$ at stress levels $x_1$ and $x_2,$ respectively,
 the  effect on the lifetime distribution of increasing  the stress from a level $x_1$ to $x_2$ at a fixed time $\tau,$  is mathematically shown as a translation of the distribution function  from $F_2(t)$ to $F_2(t - h),$ where the shifting time $h$ is such that the distribution is continuous at the time of stress change,
$F_2(\tau + h) = F_1(\tau).$
The shifting time $h$ features the accumulated damage at the preceding stress level. If the assumed lifetime distributions $F_1(t)$ and $F_2(t)$  belong to the same parametric family with scale parameter depending on the stress level, $\theta,$ we readily find that the shifting time $h$  must satisfy the relation
\begin{equation}\label{eq:shifting}
	h = \frac{\theta_2}{\theta_1}\tau - \tau.
\end{equation}
%
The CE is one of the most prominent and well-established cumulative damage model for step-stress modelling and so we will adopt this approach here for relating lifetime distributions at different steps of the stress-experiments.

Finally, in many practical situations it may be of interest to reduce the parameter space to some values satisfying a general constraint of the form
\begin{equation} \label{eq:restrictedspace}
	\Theta_0 = \{\boldsymbol{\theta} : \boldsymbol{m}(\boldsymbol{\theta}) = \boldsymbol{0}_r \},
\end{equation}
where $\boldsymbol{m} : \Theta \rightarrow \mathbb{R}^r$ ($r\leq 3$) is a differentiable function and its derivative
$\boldsymbol{M}(\boldsymbol{\theta}) = \frac{\partial  \boldsymbol{m}^T(\boldsymbol{\theta}) }{\partial \boldsymbol{\theta}} $
 exists and is continuous in $\boldsymbol{\theta}$ and has dim($\boldsymbol{M}(\boldsymbol{\theta})$) $= 3 \times r$ and rank($\boldsymbol{M}(\boldsymbol{\theta})$)$ = r$.
For example, some test statistics such as Rao tests, Lagrange multipliers tests or likelihood ratio test statistics require of restricted estimators.
 \cite{Basu2018} presented the RMDPDE in general statistical models and defined robust divergence-based tests based on this restricted estimator for testing composite null hypotheses.
In this paper we consider the RMDPDE for the step-stress ALT model with interval-censored data and we develop robust Rao-type test statistics based on the RMDPDEs.
 
The rest of paper is divided as follows: 
Section \ref{sec:Weibull} presents the Weibull lifetime distribution under step-stress testing. 
Section \ref{sec:MDPDE} introduces the RMDPDE for step-stress ALT under Weibull lifetime distributions and interval-censored data. 
Section \ref{sec:IF} theoretically analyzes the robustness of the RMDPDE through the derivation of its influence function (IF).
In Section \ref{sec:Rao}, a robust generalization of the Rao test based on the RMDPDE is presented, and explicit expressions for the test under composite null hypothesis are derived. The robustness properties and efficiency of the RMDPDE and corresponding Rao-type tests are empirically examined in Section \ref{sec:simulatio} through Monte Carlo simulation based on a real-life experiment on solar lighting devices.

\section{The Weibull lifetime distribution \label{sec:Weibull} }

The Weibull distribution is widely used to model lifetime data, such as failure times of devices, survival times of patients or durations of events. 
For example, in reliability engineering it is frequently used for analyzing the reliability and maintainability of products or systems.
It is a simple distribution but can capture different types of failure behaviours such as constant, increasing or decreasing failure rates.

The Weibull  distribution is defined by two distribution parameters; a shape parameter $\eta$ determining the curvature of the distribution and a scale parameter $\lambda$ determining its location and spread. 
The Weibull cumulative distribution function with parameters $(\eta, \lambda)$ is given by
$$ F(t) = 1-\exp\left(-\left(\frac{t}{\lambda} \right)^\eta\right).$$
Note that the Weibull family includes the exponential distribution for $\eta=1$, which is a more simple distribution suitable for modelling lifetime data with constant failure rates.
We will assume that the scale parameter $\lambda$ of the Weibull distribution depends on the stress level $x$ to which unit are subjected trough a log-linear relationship of the form
$$ \lambda = \exp(a_0 + a_1 x), \hspace{0.3cm} a_0 \in \mathbb{R}, a_1 \in \mathbb{R}^{-}. $$
The above log-linear relationship accommodates some physical models such as the Arrhenius law for the temperature dependence of reaction rates or the Inverse Power law  for  non-thermal accelerated stresses in mechanical systems.
Moreover, assuming a link function relating the scale parameter with the stress avoids the estimation of different scale parameters, as it suffices to estimate only two parameters, $a_0$ and $a_1,$ regardless of the number of stress changes considered in the experiment. 
The shape parameter $\eta$ is assumed to be constant and independent of the stress level. This homogeneity assumption is a very convenient for CE modelling, as the shifting time at each change of stress has a explicit expression given in Equation (\ref{eq:shifting}).
We denote $\boldsymbol{\theta} = ( a_0,a_1, \eta)^T \in \Theta = \mathbb{R}\times \mathbb{R}^{-}\times \mathbb{R}^+$ the vector of model parameters to be estimated.

We now discuss the step-stress model under the Weibull distribution. Let us consider a step-stress experiment with $k$ stress levels, $x_1 < \cdots < x_k,$ increased at the fixed times of stress change $\tau_1 < \cdots < \tau_{k-1},$ and let $\tau_k$ denote the experiment termination time. The proposed model is Type I-censored because the end of the experiment is fixed in advance.
Further, let us assume that the random variable describing the time until failure under a constant stress follows a Weibull distribution with  shape parameter $\eta$ and scale parameter $\lambda$ depending on the stress level.
Under the CE model, the  Weibull cumulative distribution function of the lifetime, $T,$ is given by 
$$ F_T (t) = F_{i}(t + h_{i-1}) = 1-\exp\left(-\left( \frac{t + h_{i-1} }{\lambda_i}\right) ^\eta\right), \hspace{0.3cm} i=1,...,k$$
where $h_{i-1}$ denotes the shifting time at time of stress change $\tau_i,$ and $h_0$ is defined as $h_0=0$ for notational convenience.


Great efforts have been made to model the relationships between hazard and environmental stress conditions.
The hazard function of a Weibull distributions with shape and scale parameters $\eta$ and $\lambda$, respectively, is given by
 $$ h(t) = \frac{\eta}{\lambda} \left(\frac{t}{\lambda}\right)^{\eta - 1} \hspace{.3in} t \ge 0; \eta > 0.$$  Then,  the hazard function of a Weibull distribution is increasing or decreasing depending on it shape parameter.
 Consequently, the hazard function of the lifetime $T$ is a piecewise function of  Weibull hazards, given by
 $$ h_T(t) = h_i(t+ h_{i-1}) = \frac{\eta}{\lambda_i} \left(\frac{t+ h_{i-1}}{\lambda_i}\right)^{\eta - 1} \hspace{.3in} t \ge 0; \eta > 0.$$ 
 When  $\eta > 1$, the hazard function is increasing, indicating an increasing failure rate over time, such as those caused by wear-out or ageing. 
 Conversely, when $\eta< 1$ the hazard function is decreasing, indicating a decreasing failure rate over time, adequate for modelling failures caused by infant mortality or learning effects.
Additionally, once the lifetime distribution is known, we can compute some lifetime characteristics of interest such as the mean time to failure (MTTF) of the product, given for Weibull distributions as $$MTTF = \lambda \cdot \Gamma\left(1+\frac{1}{\eta}\right).$$
In some reliability analyses, these lifetime characteristics are the target of interest, rather than the explicit expression of the reliability function.

In interval-monitored experiments exact failure times are not available, but rather failure counts are be recorded at some fixed inspection times. 
This set of inspections times usually includes all times of stress change, but additional intermediate inspection times are also possible.
The grid of inspection times, including all the times of stress change $\tau_1,...,\tau_k,$ is denoted by $0 = t_0 < \cdots < t_L,$ with $L$ the number of inspection during the experiment. With the previous notation, the probability of failure within the $j-$interval, $(t_{j-1}, t_j]$ is given by
\begin{equation} \label{eq:thprob}
	\pi_j(\boldsymbol{\theta}) = F_T (t_j) - F_T (t_{j-1})
\end{equation}
and it would depend on the stress to which units are subject between two subsequent inspection times, $t_{j-1}$ and $t_j.$ Since all times of stress change are inspections times, we can assume that the stress level is constant within an interval.

The observed data are counts of events and so can be modelled as a multinomial sample with $L+1$ events or categories; the first $L$ events stand for failing at each of the pre-defined intervals and the $L+1$-th event corresponds to  surviving at the end of the experiment. 
The theoretical probabilities associated with these events are  defined in Equation (\ref{eq:thprob}) for $j=1,...,L$ and $\pi_{L+1}(\boldsymbol{\theta}) = 1- F_T (t_L).$ Then,  the log-likelihood function of the step-stress model with interval-monitoring, given the grouped sample of failure counts, $(n_1,...,n_{L+1}),$ is given by
\begin{equation}\label{eq:loglikelihood}
	\begin{aligned}
		L(\boldsymbol{\theta} ; n_1,...,n_{L+1}) &= \log\left(\frac{(\sum_{j=1}^{L+1}n_j)!}{\prod_{j=1}^{L+1}n_j!} \cdot \prod_{j=1}^{L+1} \pi_j(\boldsymbol{\theta})^{n_j}\right)\\
		&= \log\left((\sum_{j=1}^{L+1}n_j)!\right) - \sum_{j=1}^{L+1}\log(n_j!) + \sum_{j=1}^{L+1} {n_j}\log(\pi_j(\boldsymbol{\theta})).
	\end{aligned}
\end{equation}
The MLE is defined as the maximizer (or minimizer) of the likelihood (or negative likelihood) function.
As the likelihood  function only depends on the model parameters in the last term,  the MLE can be practically computed as
$$\widehat{\boldsymbol{\theta}}_{MLE}= \operatorname{arg} \operatorname{ min} \left[-\sum_{j=1}^{L+1} n_j \log(\pi_j(\boldsymbol{\theta}))\right].$$
The MLE enjoys good asymptotic properties and it has been widely studied in the literature of interval-monitored step-stress ALTs. Also, EM-algorithms have been developing for efficiently estimating this MLE. However, the MLE has poor robustness properties and can lead to a substantial degradation in the performance of classical tests under data contamination and model misspecification.

\section{Minimum restricted density power divergence estimator  \label{sec:MDPDE}}

Restricted estimators can be used in a variety of statistical applications. 
For example, they can be used in multivariate measurement error regression models when some variables cannot be measured accurately or  for the Poisson regression model in the presence of multicollinearity.
RMDPDE have been also applied in testing hypotheses for their good  robustness properties. 
In this section, we develop robust inferential DPD-based techniques for interval-monitoring step-stress ALT experiments under Weibull lifetime distributions  under general equality restrictions.

We consider the same experimental set-up defined in Section \ref{sec:Weibull}, with $k$ stress levels and $L$ inspection times. 
Using the same notation, $n_j$ denotes the number of failures within the interval $(t_{j-1}, t_j],$ $j=1,...,L$ and $n_{L+1}$ is the number of surviving devices after the end of the experiment. As the observed data follows a multinomial model with  $N = \sum_{j=1}^{L+1} n_j$ trails (units under test) and theoretical probability vector $\boldsymbol{\pi}(\boldsymbol{\theta}) = \left(\pi_{1}(\boldsymbol{\theta}), ..., \pi_{L+1}(\boldsymbol{\theta}) \right)$ with $\pi_{j}(\boldsymbol{\theta})$ defined as in Equation (\ref{eq:thprob}), an estimator of the theoretical probability vector is naturally defined by the empirical vector of failure frequencies, $\widehat{\boldsymbol{p}} = (\widehat{p}_1,..., \widehat{p}_{L+1})$ with $\widehat{p}_j = \frac{n_j}{N}.$ 

For discrete models such as the multinomial distribution, divergence-based estimators are defined as minimizers of a suitable divergence between the empirical and theoretical probability vectors. The DPD between $\widehat{\boldsymbol{p}}$ and $\boldsymbol{\pi}(\boldsymbol{\theta})$ is given, for $\beta >0,$ by
\begin{equation}\label{eq:DPDloss}
	d_{\beta}\left( \widehat{\boldsymbol{p}},\boldsymbol{\pi}\left(\boldsymbol{\theta}\right)\right)   = \sum_{j=1}^{L+1} \left(\pi_j(\boldsymbol{\theta})^{1+\beta} -\left( 1+\frac{1}{\beta}\right) \widehat{p}_j\pi_j(\boldsymbol{\theta})^{\beta}  +\frac{1}{\beta} \widehat{p}_j^{\beta+1} \right).
\end{equation}
and the MDPDE is computed as the minimizer of Equation (\ref{eq:DPDloss}). 
Our choice of DPD is motivated by the fact that DPD-based estimators and tests have demonstrated good robustness performance without a significant loss of efficiency across many statistical models.
The non-negative tuning parameter $\beta$ controls the trade off between robustness and asymptotic efficiency of the parameter estimates.
For $\beta=0,$ the DPD can be defined by taking continuos limits  $\beta \rightarrow 0,$ and the resulting expression coincides with the well-known  Kullback-Leibler divergence between the empirical and theoretical probability vectors,
$$d_{0} (\boldsymbol{\widehat{p}}, \boldsymbol{\pi}(\boldsymbol{\theta})) = \sum_{j=1}^{L+1}\widehat{p}_j \log\left(\frac{\widehat{p}_j}{\pi_j(\boldsymbol{\theta})}\right).$$
Note that the objective function based on the Kullback-Leibler divergence reaches its minimum at the same point that the negative-loglikelihood given in Equation (\ref{eq:loglikelihood}). Then, the minimum Kullback-Leibler divergence estimator coincides with the MLE and it justifies the classical likelihood-based estimator from the information theory approach.
Further, the DPD family  generalizes the likelihood procedure to a broader  class of estimators, including the  classical MLE as a special case.

If we consider a restricted parameter space of the form
$$\Theta_0 = \{\boldsymbol{\theta} : \boldsymbol{m}(\boldsymbol{\theta}) = \boldsymbol{0}_r \},$$
where $\boldsymbol{m} : \Theta \rightarrow \mathbb{R}^r,$ $r\leq 3,$ is a differentiable function,
we can readily defined the RMDPDE as
\begin{equation}\label{eq:RMDPDE}
			\widetilde{\boldsymbol{\theta}}_\beta = \operatorname{arg } \operatorname{min}_{\boldsymbol{\theta} \in \Theta_0} d_{\beta}\left( \widehat{\boldsymbol{p}},\boldsymbol{\pi}\left(\boldsymbol{\theta}\right)\right) .
\end{equation}

The above constrained minimization problem can be solved using Lagrange multipliers. Naturally, the asymptotic properties of the resulting estimators and the estimating equations would be affected by the imposed constraints. 
The next theorem states the estimating equations of the RMDPDE.

\begin{theorem}\label{thm:estimatingequations}
	The RMDPDE of the interval-censored step-stress ALT model under Weibull lifetime distributions must satisfy the following system of $3+ r$ equations
	\begin{align}
		\boldsymbol{W}^T(\boldsymbol{\theta}) \boldsymbol{D}^{\beta-1}_{\boldsymbol{\pi}(\boldsymbol{\theta})} \left(\widehat{\boldsymbol{p}}- \boldsymbol{\pi}(\boldsymbol{\theta}) \right) + \boldsymbol{M}(\boldsymbol{\theta}) \boldsymbol{\lambda}_\beta &= \boldsymbol{0}_3 \label{eq:estimatingMDPDE} \\
		\boldsymbol{m}(\boldsymbol{\theta}) &= \boldsymbol{0}_r	\label{eq:estimatingRMDPDE}
	\end{align}
	where $\boldsymbol{D}_{\boldsymbol{\pi}(\boldsymbol{\theta})}$ denotes a $(L+1)\times(L+1)$ diagonal matrix with diagonal entries $\pi_j(\boldsymbol{\theta}),$ $j=1,...,L+1,$ and $\boldsymbol{W}(\boldsymbol{\theta})$ is a $(L+1) \times 3$ matrix with rows
	$
	\boldsymbol{w}_j =  \boldsymbol{z}_j-\boldsymbol{z}_{j-1},
	$
	where
	\begin{align}
		\label{eq:zj} \boldsymbol{z}_j  &= g_T(t_j)\begin{pmatrix}
			-(t_j+h_{i-1})\\
			-(t_j+h_{i-1})x_i + h_{i-1}^\ast \\
			\log\left(\frac{t_j+h_{i-1}}{\lambda_i}\right)\frac{t_j+h_{i-1}}{\eta}
		\end{pmatrix}, \hspace{0.3cm} j = 1,...,L\\ 
		h_{i}^\ast &= h_{i}x_{i+1}+ \alpha_{i+1}\sum_{k=0}^{i-1}\left(\frac{x_{i+1-k}}{\lambda_{i+1-k}} - \frac{x_{i-k}}{\lambda_{i-k}} \right)\tau_{i-k} \label{aast},\hspace{0.3cm} i=1,...,k-1.
	\end{align}
	 $ \boldsymbol{z}_{-1} = \boldsymbol{z}_{L+1} = \boldsymbol{0},$ $i$ is the stress level at which the units are tested before the $j-$th inspection time,
	and $\boldsymbol{\lambda}_\beta$ denotes the $r$-dimensional vector of Lagrange multipliers.
\end{theorem}

Setting $\boldsymbol{\lambda}_\beta = \boldsymbol{0}_r$ in Equation (\ref{eq:estimatingMDPDE}) yields the same estimating equations of the (unrestricted) MDPDE obtained in \cite{BalaJae2022b}, and the function 
\begin{equation} \label{eq:score}
	\boldsymbol{U}_\beta(\boldsymbol{\theta}) = \boldsymbol{W}^T \boldsymbol{D}^{\beta-1}_{\boldsymbol{\pi}(\boldsymbol{\theta})} \left(\widehat{\boldsymbol{p}}- \boldsymbol{\pi}(\boldsymbol{\theta})\right) 
\end{equation}
 is called the $\beta$-score function of the model.
 Additionally, if the matrix $\boldsymbol{M}(\boldsymbol{\theta})$ is invertible, the vector of Lagrange multipliers can be obtained explicitly as a function of $\boldsymbol{\theta}$ by solving  Equation (\ref{eq:estimatingMDPDE}) in $\boldsymbol{\lambda}_\beta.$ However, in many situations the matrix $\boldsymbol{M}(\boldsymbol{\theta})$ will not be invertible and  the vector of Lagrange multipliers will have to be computed numerically.

The next theorem establishes the main asymptotic properties of the RMDPDE under the step-stress set-up.

\begin{theorem}
	Assume that all defined matrices exists finitely and suppose  that the true distribution belongs to the model with true parameter $\boldsymbol{\theta}_0 \in \Theta.$ 
	Then, the RMDPDE of the interval-censored step-stress ALT model under Weibull lifetime distributions satisfying the constraints $\boldsymbol{m}(\boldsymbol{\theta}) = \boldsymbol{0}_r$ has the following properties.
	\begin{enumerate}
		\item The RMDPDE estimating equation has a consistent sequence of roots $\widetilde{\boldsymbol{\theta}}_\beta$ such that $\widetilde{\boldsymbol{\theta}}_\beta \xrightarrow{P} \boldsymbol{\theta}_0.$
		\item The null distribution of the RMDPDE is given by
			$$\sqrt{N}\left(\widetilde{\boldsymbol{\theta}}_\beta-\boldsymbol{\theta}_0 \right) \xrightarrow[N \rightarrow \infty]{L}\mathcal{N}\left(\boldsymbol{0}, \boldsymbol{\Sigma}_\beta(\boldsymbol{\theta}_0)\right)$$
		where 
		\begin{equation}\label{eq:SigmaPQmatrices}
			\begin{aligned}
				\boldsymbol{\Sigma}_\beta(\boldsymbol{\theta}_0) &= \boldsymbol{P}_\beta(\boldsymbol{\theta}_0) \boldsymbol{K}_\beta(\boldsymbol{\theta}_0) \boldsymbol{P}_\beta(\boldsymbol{\theta}_0)^T,\\
				\boldsymbol{P}_\beta(\boldsymbol{\theta}_0) &=  \boldsymbol{J}_\beta(\boldsymbol{\theta}_0)^{-1}- \boldsymbol{Q}_\beta(\boldsymbol{\theta}_0)\boldsymbol{M}^T(\boldsymbol{\theta}_0) \boldsymbol{J}_\beta(\boldsymbol{\theta}_0)^{-1},\\
				\boldsymbol{Q}_\beta(\boldsymbol{\theta}_0) &=  \boldsymbol{J}_\beta(\boldsymbol{\theta}_0)^{-1}\boldsymbol{M}(\boldsymbol{\theta}_0)[ \boldsymbol{M}^T(\boldsymbol{\theta}_0) \boldsymbol{J}_\beta(\boldsymbol{\theta}_0)^{-1}\boldsymbol{M}(\boldsymbol{\theta}_0)]^{-1},
			\end{aligned}
		\end{equation} 
		with \begin{equation} \label{eq:JK}
			\boldsymbol{J}_\beta(\boldsymbol{\theta}_0) = \boldsymbol{W}^T \boldsymbol{D}_{\boldsymbol{\pi}(\boldsymbol{\theta_0})}^{\beta-1} \boldsymbol{W},
			\hspace{0.3cm}  \hspace{0.3cm}
			\boldsymbol{K}_\beta(\boldsymbol{\theta}_0) = \boldsymbol{W}^T \left(\boldsymbol{D}_{\boldsymbol{\pi}(\boldsymbol{\theta_0})}^{2\beta-1}-\boldsymbol{\pi}(\boldsymbol{\theta}_0)^{\beta}\boldsymbol{\pi}(\boldsymbol{\theta}_0)^{\beta T}\right) \boldsymbol{W},
		\end{equation}
		$\boldsymbol{D}_{\boldsymbol{\pi}(\boldsymbol{\theta_0})}$ denotes the diagonal matrix with entries $\pi_j(\boldsymbol{\theta_0}),$ $j=1,...,L+1,$ and $\boldsymbol{\pi}(\boldsymbol{\theta}_0)^{\beta}$ denotes the vector with components $\pi_j(\boldsymbol{\theta}_0)^{\beta}.$
	\end{enumerate}
\end{theorem}
\begin{proof}
		Following \cite{Basu2022}, the asymptotic distribution of the RMDPDE for general statistical models is given by
			$$\sqrt{N}\left(\widetilde{\boldsymbol{\theta}}_\beta-\boldsymbol{\theta}_0 \right) \xrightarrow[N \rightarrow \infty]{L}\mathcal{N}\left(\boldsymbol{0}, \boldsymbol{\Sigma}_\beta(\boldsymbol{\theta}_0)\right)$$
		with the matrices $\boldsymbol{\Sigma}_\beta(\boldsymbol{\theta}_0),	\boldsymbol{P}_\beta(\boldsymbol{\theta}_0)$ and $\boldsymbol{Q}_\beta(\boldsymbol{\theta}_0)$ given as in (\ref{eq:SigmaPQmatrices}), and 		
		matrices $\boldsymbol{J}_\beta(\boldsymbol{\theta}) $ and $\boldsymbol{K}_\beta(\boldsymbol{\theta}) $  given by
		\begin{align*}
			\boldsymbol{J}_\beta(\boldsymbol{\theta}) &= \sum_{j=1}^{L+1} \boldsymbol{u}_j \boldsymbol{u}_j^T \pi_j(\boldsymbol{\theta})^{\beta+1}\\
			\boldsymbol{K}_\beta(\boldsymbol{\theta}) &= \sum_{j=1}^{L+1} \boldsymbol{u}_j \boldsymbol{u}_j^T \pi_j(\boldsymbol{\theta})^{2\beta+1} - \left( \sum_{j=1}^{L+1} \boldsymbol{u}_j  \pi_j(\boldsymbol{\theta})^{\beta+1}\right) \left( \sum_{j=1}^{L+1} \boldsymbol{u}_j  \pi_j(\boldsymbol{\theta})^{\beta+1}\right)^T,\\
		\end{align*}
		where 
		$$\boldsymbol{u}_j(\boldsymbol{\theta}) = \frac{\partial \log(\pi_j(\boldsymbol{\theta}))}{\partial \boldsymbol{\theta}}.$$
		For the step-stress ALT model under interval-censored data, we have that
		$$\boldsymbol{u}_j(\boldsymbol{\theta})  = \frac{1}{\pi_j(\boldsymbol{\theta})}\frac{\partial \pi_j(\boldsymbol{\theta})}{\partial \boldsymbol{\theta}} = \frac{\boldsymbol{w}_j}{\pi_j(\boldsymbol{\theta})}$$
		and hence, we can write the matrices $\boldsymbol{J}_\beta(\boldsymbol{\theta})$ and $\boldsymbol{K}_\beta(\boldsymbol{\theta})$ as
		\begin{align*}
			\boldsymbol{J}_\beta(\boldsymbol{\theta}) &= \sum_{j=1}^{L+1} \boldsymbol{w}_j \boldsymbol{w}_j^T \pi_j(\boldsymbol{\theta})^{\beta-1} = \boldsymbol{W}^T D_{\boldsymbol{\pi}(\boldsymbol{\theta})}^{\beta-1} \boldsymbol{W}\\
			\boldsymbol{K}_\beta(\boldsymbol{\theta}) &= \sum_{j=1}^{L+1} \boldsymbol{w}_j\boldsymbol{w}_j^T \pi_j(\boldsymbol{\theta})^{2\beta-1} - \left( \sum_{j=1}^{L+1} \boldsymbol{w}_j \pi_j(\boldsymbol{\theta})^{\beta}\right) \left( \sum_{j=1}^{L+1} \boldsymbol{w}_j \pi_j(\boldsymbol{\theta})^{\beta}\right)^T = \boldsymbol{W}^T \left(D_{\boldsymbol{\pi}(\boldsymbol{\theta})}^{2\beta-1}-\boldsymbol{\pi}(\boldsymbol{\theta})^{\beta}\boldsymbol{\pi}(\boldsymbol{\theta})^{\beta T}\right) \boldsymbol{W}.\\
		\end{align*}
		and then the result holds.
\end{proof}

The above theorem proves the consistency of the  RMDPDE under  the restricted parameter space. Also, a consistent estimator of the variance-covariance matrix  $\boldsymbol{\Sigma}_\beta(\boldsymbol{\theta}_0)$ and be readily obtained as $\boldsymbol{\Sigma}_\beta(\widetilde{\boldsymbol{\theta}}_\beta)$ and consequently univariate approximate confidence intervals for the model parameters can be computed using Wald confidence intervals of the form
$$ IC_\alpha(\theta_i) = \left\{\theta_i  : \widetilde{\theta}_{\beta, i} -   \frac{\Sigma(\widetilde{\boldsymbol{\theta}}_\beta)_{ii}}{\sqrt{N}} z_{\alpha/2} < \theta_i < \widehat{\theta}_{\beta, i} +   \frac{\Sigma(\widetilde{\boldsymbol{\theta}}_\beta)_{ii}}{\sqrt{N}} z_{\alpha/2} \right\}, \hspace{0.3cm} i=1,2,3$$
where $\Sigma(\widetilde{\boldsymbol{\theta}}_\beta)_{ii}$ denotes the $i-$th diagonal entry of the covariance matrix in (\ref{eq:SigmaPQmatrices}), $ \widetilde{\theta}_{\beta, i} $ denotes the  $i-$th component of the parameter vector $\boldsymbol{\theta}$ and $z_{ \alpha/2}$ denotes the upper $\alpha/2$ quantile of a standard normal distribution.
For the multivariate case, the approximate elliptical confidence region for the true model parameter $\boldsymbol{\theta}$ for a $(1-\alpha)$ confidence level is given by,
$$ \mathcal{R}_\alpha(\boldsymbol{\theta}) = \{ \boldsymbol{\theta} \in \Theta : N\left(\widetilde{\boldsymbol{\theta}}_\beta-\boldsymbol{\theta} \right)^T  \boldsymbol{\Sigma}_\beta(\widetilde{\boldsymbol{\theta}}_\beta)^{-1} \left(\widetilde{\boldsymbol{\theta}}_\beta-\boldsymbol{\theta} \right) \leq \chi^2_{3, \alpha}\} $$
where $\chi^2_{3, \alpha}$ denotes the upper $\alpha$ quantile of a chi-square distribution with $3$ degrees of freedom.
Besides, if one is interested in inferring a distribution lifetime characteristic that depends explicitly on the model parameters under normal operating conditions, such as the MTTF or the hazard function at a certain mission time,  a consistent robust estimator can be directly obtained by plugging-in the MDPDE on the corresponding expression. 


\section{Robustness analysis \label{sec:IF}}

The robustness of the MDPDE has been studied 
in multinomial regression models in \cite{Castilla2021, Calvino}. For continuous variables, the influence function (IF) of an estimator measures the effect on the estimate of a single point perturbation of the sample. Robust estimators are supposed to be less influenced by contamination than non-robust estimator, and so their IF should be lower than the IF of a non-robust estimator. Also, it is a desirable property for a robust estimator that the influence of a point disturbance is bounded, so it IF should also be bounded. 

The IF of a estimator is computed in terms of its associated statistical functional. 
Let $F_{\boldsymbol{\theta}}$ and  $G$ be the assumed and true lifetime cumulative distributions underlying and consider a point perturbation $t_0.$ The IF of an estimator with associated functional $\boldsymbol{T}(G)$ is computed as the  derivative of the functional with respect to a smooth deviation of the true distribution at a point perturbation,  $G_\varepsilon = (1-\varepsilon)\boldsymbol{G} + \varepsilon\Delta_{t_0},$ where $\Delta_{t_0}$ denotes the cumulative distribution function of a degenerate variable at $t_0$ and $\varepsilon$ is the contamination proportion. That is, 
\begin{equation} \label{eq:IFmath}
	\text{IF}\left(t_0, \boldsymbol{T}, G\right) = \lim_{\varepsilon \rightarrow 0}\frac{\boldsymbol{T}(G_\varepsilon)- \boldsymbol{T}(G)}{\varepsilon} = \frac{\partial \boldsymbol{T}(G_\varepsilon)}{\partial \varepsilon}\bigg|_{\varepsilon = 0},
\end{equation} 

For the interval-monitored step-stress ALT model with Weibull lifetime distributions, \cite{BalaJae2023a} derived the IF of the MDPDE at the point perturbation $t_0$ and true parameter value $\boldsymbol{\theta}_0$ as
\begin{equation}\label{eq:IF}
		\text{IF}\left(t_0, \boldsymbol{T}_\beta, F_{\boldsymbol{\theta}_0} \right) = \boldsymbol{J}_\beta^{-1}(\boldsymbol{\theta}_0) \boldsymbol{W}^T(\boldsymbol{\theta}_0) \boldsymbol{D}_{\boldsymbol{\pi}(\boldsymbol{\theta}_0)}^{\beta-1}\left(-\boldsymbol{\pi}(\boldsymbol{\theta}_0)+\delta_{t_0} \right).
\end{equation}
where $ \boldsymbol{T}_\beta$ is the functional associated with the MDPDE with tuning parameter $\beta,$ $\boldsymbol{\pi}(\boldsymbol{\theta})$ and $\boldsymbol{g}$ are the  multinomial probability vectors associated  to the distributions $F_{\boldsymbol{\theta}}$ and $G,$ respectively, as defined in Equation (\ref{eq:thprob}). Here $\delta_{t_0}$ denotes a degenerate multinomial probability vector with all probability concentrated at the interval containing the point $t_0.$
The matrices  $\boldsymbol{W}(\boldsymbol{\theta}),$ $\boldsymbol{J}_\beta(\boldsymbol{\theta})$ and $\boldsymbol{D}_{\boldsymbol{\pi}(\boldsymbol{\theta})}$ are as defined in Theorem \ref{thm:estimatingequations} and Equation (\ref{eq:JK}).

 Here, the stress levels and inspections times are the continuous covariates of the model, taking values in a infinite set. Then, studying the boundedness of the IF in its classical approach only makes sense for these continuous covariates.  \cite{BalaJae2022b} proved the boundedness of the IF of the MDPDE for large stress levels and inspection times.
 On the other hand, since the multinomial variable has a finite support, multinomial samples can only take  values in a finite set. 
 Hence, the previous definition of a point contamination taking values towards infinity are intuitively not meaningful for representing multinomial outliers.
 Outlier points within the support of the multinomial  variable are  produced by an abnormal failure counts within a cell, either for early or late failures  of a set of units. These multinomial outliers are generated by a point perturbation of the lifetime distribution $t_0,$ fixed during the time of experimentation.
 The resulting degenerate distribution of the multinomial model  concentrates all probability at one cell and so,
 varying over all possible such outlier cells, we get all possible counting errors. The influence function given in Equation (\ref{eq:IF}) then represents the bias in the estimator caused by infinitesimal amount of such errors, and robustness should be better examined it terms of gross error sensitivity of the functional, as discussed in \cite{CastillaJaenada2022}.
 This view of considering outliers in multinomial sampling as classification errors is in line with the general literature on robust analysis of
 categorical data, including different types of logistic regressions having finite supports for the model densities 
 (\cite{johnson1985}, 
  \cite{croux2003}, 
  \cite{bondell2008}, 
  \cite{basu2017}, 
  \cite{Castilla2021}).

 

For restricted estimators, \cite{Ghosh2015} derived the IF of the RMDPDE in general statistical models.  
They defined the functional $\widetilde{\boldsymbol{T}}_\beta$ associated to the RMDPDE  as a solution of the estimating equations of the (unconstrained) MDPDE given in (\ref{eq:estimatingMDPDE}) over the restricted subspace $\Theta_0$ given in Equation (\ref{eq:restrictedspace}). The existence of such solution is guaranteed by the Implicit Function Theorem.
 Hence, the IF of the RMDPDE at a contamination point $t_0$ and the model distribution with true parameter value $\boldsymbol{\theta}_0$, 
  must simultaneously verify the expression of the IF of the (unrestricted) MDPDE 
 and the subspace constraint, given by
 $\boldsymbol{m}(\widetilde{\boldsymbol{T}}_\beta) = 0.$
 Differentiating on the subspace constraint, we get an additional condition for the IF,
 $$\boldsymbol{M}^T(\boldsymbol{\theta}_0) \text{IF}\left(\boldsymbol{n}, \widetilde{\boldsymbol{T}}_\beta, F_{\boldsymbol{\theta}_0}\right) = 0$$
 and therefore, combining both equations, we get that the IF of the RMDDE is given by
 \begin{equation*}
 	\begin{pmatrix}
 		\boldsymbol{J}_\beta(\boldsymbol{\theta}_0)\\
 		\boldsymbol{M}^T (\boldsymbol{\theta}_0)
 	\end{pmatrix} \text{IF}\left(\boldsymbol{n}, \widetilde{\boldsymbol{T}}_\beta, F_{\boldsymbol{\theta}_0} \right)
 	= \begin{pmatrix}
 		\boldsymbol{W}^T \boldsymbol{D}_{\boldsymbol{\pi}(\boldsymbol{\theta}_0)}^{\beta-1}\left(-\boldsymbol{\pi}(\boldsymbol{\theta}_0)+\delta_{t_0} \right) \\
 		0
 	\end{pmatrix}.
 \end{equation*}
Finally, multiplying both terms by $ \left(\boldsymbol{J}_\beta(\boldsymbol{\theta})^T, \boldsymbol{M}(\boldsymbol{\theta}_0) \right)$ and inverting in both sizes of the equation, the expression of the IF of the restricted MDPDE is given by
 \begin{equation}\label{eq:IFrestricted}
 	\text{IF}\left(\boldsymbol{n}, \widetilde{\boldsymbol{T}}_\beta, F_{\boldsymbol{\theta}_0}\right)
 	= \left(\boldsymbol{J}_\beta(\boldsymbol{\theta}_0)^T\boldsymbol{J}_\beta(\boldsymbol{\theta}_0) + \boldsymbol{M}(\boldsymbol{\theta}_0)\boldsymbol{M}(\boldsymbol{\theta}_0)^T \right)^{-1}
 	\boldsymbol{J}_\beta(\boldsymbol{\theta}_0)^T
 	\boldsymbol{W}^T \boldsymbol{D}_{\boldsymbol{\pi}(\boldsymbol{\theta}_0)}^{\beta-1}\left(-\boldsymbol{\pi}(\boldsymbol{\theta}_0)+\delta_{t_0} \right).
 \end{equation}
The matrix $\left(\boldsymbol{J}_\beta(\boldsymbol{\theta}_0)^T\boldsymbol{J}_\beta(\boldsymbol{\theta}_0) + \boldsymbol{M}(\boldsymbol{\theta}_0)\boldsymbol{M}(\boldsymbol{\theta}_0)^T\right)^{-1}\boldsymbol{J}_\beta(\boldsymbol{\theta}_0)^T$ is typically assumed to be bounded. Then, the robustness of the RMDPDE depends  on the boundedness of the IF of the (unrestricted) MDPDE. Following the previous discussion, the RMDPDE is robust for positives values of $\beta$ whereas the restricted MLE (corresponding to $\beta=0$) lacks of robustness against outlying  stress levels or inspection times data.
 
 

\section{Robust Rao-type test statistics \label{sec:Rao}}
We consider a testing problem with composite null hypothesis of the form
\begin{equation} \label{eq:null}
	\operatorname{H}_0 : \boldsymbol{m}(\boldsymbol{\theta}) = \boldsymbol{0} \text{ vs } \operatorname{H}_1 : \boldsymbol{m}(\boldsymbol{\theta}) \neq \boldsymbol{0}.
\end{equation}

 \cite{Rao1948} introduced the Rao test (or the score test) as an alternative to the well-known likelihood ratio and Wald test. Rao's test uses the restricted MLE and may outperform the Wald and likelihood ratio tests in some situations, although there is no better statistics overall.
All these classical tests are based on the MLE, so they will inherit the lack of robustness of the MLE.
To overcome the robustness drawback with a small loss of efficiency,
\cite{Basu2018} developed  robust parametric tests for composite hypotheses based on RMDPDEs. Later,
\cite{Basu2022} generalizes the Rao test using the RMDPD for general statistical models, and
\cite{JaenadaMirandaPardo2022} studied Rényi pseudodistance-based restricted estimators for the same purpose. 
For interval-censored step-stress ALTs, \cite{BalaJae2023a} developed RMDPDE under exponential lifetimes, and derived explicit expressions of the Rao-type tests for the exponential lifetime distribution and linear null hypothesis.
In the same spirit, we develop here a robust generalization of the Rao test for interval-monitoring step-stress ALTs but assuming Weibull lifetime distributions.

The Rao-type test statistics use the $\beta$-score function $\boldsymbol{U}_\beta(\boldsymbol{\theta})$ defined in Equation (\ref{eq:score}) and it asymptotic distribution.
Simple calculations show that $\mathbb{E}[\boldsymbol{U}_\beta(\widetilde{\boldsymbol{\theta}}_\beta)] = 0$ (as the parameter must annul the score function)	and  that it covariance matrix	is given by $\text{Cov}[\boldsymbol{U}_\beta(\widetilde{\boldsymbol{\theta}}_\beta)] = \frac{1}{N} \boldsymbol{K}_\beta(\boldsymbol{\theta}_0)$ with $\boldsymbol{K}_\beta(\boldsymbol{\theta})$ defined as in Equation (\ref{eq:JK}) and $\boldsymbol{\theta}_0$ the true value of the parameter. By the central limit theorem (CLT) we can easily establish the distribution of the $\beta-$score as follows,
$$\sqrt{N}\boldsymbol{U}(\widetilde{\boldsymbol{\theta}}_\beta) \xrightarrow[N \rightarrow \infty]{L}\mathcal{N}\left(\boldsymbol{0},  \boldsymbol{K}_\beta(\boldsymbol{\theta}_0)\right). $$

\begin{definition}
The Rao-type statistics based on the RMDPDE for testing the composite null hypothesis defined (\ref{eq:null}), $\widetilde{\boldsymbol{\theta}}^\beta,$ is given by
		\begin{equation}\label{eq:raotest}
			R_{\beta}(\widetilde{\boldsymbol{\theta}}^\beta) = N\boldsymbol{U}_{\beta}(\widetilde{\boldsymbol{\theta}}^\beta)^T \boldsymbol{Q}_{\beta}(\widetilde{\boldsymbol{\theta}}^\beta)\left[ \boldsymbol{Q}_{\beta}(\widetilde{\boldsymbol{\theta}}^\beta)^T \boldsymbol{K}_{\beta}(\widetilde{\boldsymbol{\theta}}^\beta)
			\boldsymbol{Q}_{\beta}(\widetilde{\boldsymbol{\theta}}^\beta)\right]^{-1}
			\boldsymbol{Q}_{\beta}(\widetilde{\boldsymbol{\theta}}^\beta)^T \boldsymbol{U}_{\beta,N}(\widetilde{\boldsymbol{\theta}}^\beta),
		\end{equation}
		where matrices $\boldsymbol{K}_{\beta}(\boldsymbol{\theta})$, $\boldsymbol{Q}_{\beta}(\boldsymbol{\theta})$
		and $\boldsymbol{U}_{\beta}(\boldsymbol{\theta}),$  are defined in (\ref{eq:JK}), (\ref{eq:SigmaPQmatrices}) and (\ref{eq:score}), respectively. 
\end{definition}
The above expression only depends on the restricted estimator, and so it is not necessary to compute unrestricted estimators.
Then, for simple null hypotheses, $\operatorname{H}_0 : \boldsymbol{\theta} = \boldsymbol{\theta}_0,$ no estimator is needed for defining the test, which is one of the most important advantages of the test. 
Also, if $\beta=0,$ we recover the classical Rao test statistic for the interval-monitoring step-stress ALT model.		
\begin{theorem}
	Under the null hypothesis given in Equation (\ref{eq:null}), the Rao-type test statistic defined in (\ref{eq:raotest}) converges to a chi-square distribution with $r$ degrees of freedom.
\end{theorem}
	\begin{proof}
	The MDPDE restricted to the null hypothesis (\ref{eq:null}), $\widetilde{\boldsymbol{\theta}}^\beta,$ 
	is defined as the minimum of the DPD loss restricted to the condition $\boldsymbol{m}(\boldsymbol{\theta}) =d$. Then, it satisfies the restricted equations
	\begin{equation}\label{eq:restrictedestimating}
		\boldsymbol{U}_{\beta}(\widetilde{\boldsymbol{\theta}}^\beta) + \boldsymbol{M}(\widetilde{\boldsymbol{\theta}}^\beta)\boldsymbol{\lambda}_\beta = \boldsymbol{0},
	\end{equation}
	for some vector $\boldsymbol{\lambda}_\beta$ of Lagrangian multipliers. 
	Writing the $\beta$-score as
	$\boldsymbol{U}_{\beta}(\widetilde{\boldsymbol{\theta}}^\beta) = -\boldsymbol{M}(\boldsymbol{\theta})\boldsymbol{\lambda}
	$
 	and multiplying its transpose by the matrix $ \boldsymbol{Q}_{\beta}(\widetilde{\boldsymbol{\theta}}^\beta)$ we get
	$$
	\boldsymbol{U}_{\beta}(\widetilde{\boldsymbol{\theta}}^\beta)^T \boldsymbol{Q}_{\beta}(\widetilde{\boldsymbol{\theta}}^\beta) = -\boldsymbol{\lambda}_\beta^T\boldsymbol{M}(\widetilde{\boldsymbol{\theta}}^\beta)^T\boldsymbol{Q}_{\beta}(\widetilde{\boldsymbol{\theta}}^\beta) = - \boldsymbol{\lambda}_\beta^T
	$$
	where $\boldsymbol{Q}_{\beta}(\widetilde{\boldsymbol{\theta}}^\beta)$ is defined in (\ref{eq:SigmaPQmatrices}). Now, the Rao-type test statistics defined in (\ref{eq:raotest}) can be expressed in terms of the vector of Lagrange multipliers as follows,
	\begin{equation*}
		\boldsymbol{R}_{\beta}(\boldsymbol{\theta}) = N\boldsymbol{\lambda}_\beta^T\left[ \boldsymbol{Q}_{\beta}(\boldsymbol{\theta}^\beta)^T \boldsymbol{K}_{\beta}(\boldsymbol{\theta}^\beta)
		\boldsymbol{Q}_{\beta}(\boldsymbol{\theta}^\beta)\right]^{-1}
		\boldsymbol{\lambda}_\beta.
	\end{equation*}
	Then, we need first to derive the asymptotic distribution of the vector of Lagrangian multipliers $\boldsymbol{\lambda}_\beta.$
	We  consider the second order Taylor expansion series of the $\beta$-score function $\boldsymbol{U}_{\beta}(\boldsymbol{\theta})$ around the true parameter value $\boldsymbol{\theta}_0,$ 
	$$\boldsymbol{U}_{\beta}(\widetilde{\boldsymbol{\theta}}^\beta) = \boldsymbol{U}_{\beta}(\boldsymbol{\theta}_0) + \frac{\partial \boldsymbol{U}_{\beta}(\boldsymbol{\theta})}{\partial \boldsymbol{\theta}}\bigg|_{\boldsymbol{\theta} = \boldsymbol{\theta}_0}\left(\widetilde{\boldsymbol{\theta}}^\beta- \boldsymbol{\theta}_0\right) + o\left(||\widetilde{\boldsymbol{\theta}}^\beta- \boldsymbol{\theta}_0||^2\boldsymbol{1}_2\right).$$
	On the other hand, because $\widehat{p} \xrightarrow[N\rightarrow \infty]{P} \boldsymbol{\pi}(\boldsymbol{\theta}_0)$, taking derivatives in (\ref{eq:score} ) it is not difficult to show that 
	$$\frac{\partial \boldsymbol{U}_{\beta}(\boldsymbol{\theta})}{\partial \boldsymbol{\theta}}\bigg|_{\boldsymbol{\theta} = \boldsymbol{\theta}_0} \xrightarrow[N\rightarrow \infty]{P} - \boldsymbol{J}_\beta\left(\boldsymbol{\theta}_0\right),$$
	where the matrix $\boldsymbol{J}_\beta\left(\boldsymbol{\theta}\right)$ is defined in (\ref{eq:JK}). 
	Therefore, we can approximate the $\beta$-score function at the RMDPDE by
	$$\boldsymbol{U}_{\beta}(\widetilde{\boldsymbol{\theta}}^\beta) = \boldsymbol{U}_{\beta}(\boldsymbol{\theta}_0) - \boldsymbol{J}_\beta\left(\boldsymbol{\theta}_0\right) \left(\widetilde{\boldsymbol{\theta}}^\beta- \boldsymbol{\theta}_0\right) + o\left(||\widetilde{\boldsymbol{\theta}}^\beta- \boldsymbol{\theta}_0||^2\boldsymbol{1}_2\right)+o\left( \boldsymbol{1}_2\right).$$
	Based on the estimating equations of the RMDPDE given in (\ref{eq:restrictedestimating}), we have that
	\begin{equation} \label{eq:proof1}
			\boldsymbol{U}_{\beta}(\boldsymbol{\theta}_0) - \boldsymbol{J}_\beta\left(\boldsymbol{\theta}_0\right) \left(\widetilde{\boldsymbol{\theta}}^\beta- \boldsymbol{\theta}_0\right)  + \boldsymbol{M}(\widetilde{\boldsymbol{\theta}}^\beta)\boldsymbol{\lambda}_\beta= o\left( \boldsymbol{1}_2\right).
	\end{equation}
	
	On the other hand, the second order Taylor expansion of the restriction function, $\boldsymbol{m}(\boldsymbol{\theta}),$ around  the true parameter value $\boldsymbol{\theta_0}$ at the RMDPDE yields
	$$\boldsymbol{m}(\widetilde{\boldsymbol{\theta}}^\beta) = \boldsymbol{m}(\boldsymbol{\theta}_0) + \boldsymbol{M}^T(\boldsymbol{\theta}_0)\left(\widetilde{\boldsymbol{\theta}}^\beta- \boldsymbol{\theta}_0\right) + o\left(||\widetilde{\boldsymbol{\theta}}^\beta- \boldsymbol{\theta}_0||^2\boldsymbol{1}_2\right).$$
	
	Under the null hypothesis the restriction $\boldsymbol{m}(\boldsymbol{\theta}_0)=0$ holds, and the equality is also satisfied by the restricted estimator. Then, we can write,
	\begin{equation}\label{eq:proof2}
		\boldsymbol{M}^T(\boldsymbol{\theta}_0)\left(\widetilde{\boldsymbol{\theta}}^\beta- \boldsymbol{\theta}_0\right)  = o\left(||\widetilde{\boldsymbol{\theta}}^\beta- \boldsymbol{\theta}_0||^2\boldsymbol{1}_2\right).
	\end{equation}
	Joining both equations (\ref{eq:proof1} and \ref{eq:proof2}),  we get
	\begin{equation*}
		\begin{pmatrix}
			\boldsymbol{J}_\beta(\boldsymbol{\theta}_0) & \boldsymbol{M}(\boldsymbol{\theta}_0) \\
			\boldsymbol{M}^T(\boldsymbol{\theta}_0) & \boldsymbol{0}_r
		\end{pmatrix} 
		\begin{pmatrix}
			\widetilde{\boldsymbol{\theta}}^\beta - \boldsymbol{\theta}_0\\
			-\boldsymbol{\lambda}_\beta
		\end{pmatrix} =
		\begin{pmatrix}
			\boldsymbol{U}_{\beta}(\boldsymbol{\theta}_0)\\
			 \boldsymbol{0}_{r}
		\end{pmatrix} + 	\begin{pmatrix}
			o(\boldsymbol{1}_2)\\
			o(\boldsymbol{1}_2)
		\end{pmatrix} 
	\end{equation*}
	and solving the previous equation for $(
	\widetilde{\boldsymbol{\theta}}^\beta - \boldsymbol{\theta}_0)$ and $-\boldsymbol{\lambda}_\beta$, we have that
	\begin{equation*}
		\begin{pmatrix}
			\widetilde{\boldsymbol{\theta}}^\beta - \boldsymbol{\theta}_0\\
			-\boldsymbol{\lambda}_\beta
		\end{pmatrix} =
		\begin{pmatrix}
			\boldsymbol{J}_\beta(\boldsymbol{\theta}_0) & \boldsymbol{M}(\boldsymbol{\theta}_0)\\
			\boldsymbol{M}^T(\boldsymbol{\theta}_0) & \boldsymbol{0}
		\end{pmatrix}^{-1}
		\begin{pmatrix}
			\boldsymbol{U}_{\beta}(\boldsymbol{\theta}_0)\\
			\boldsymbol{0}_r
		\end{pmatrix}  + 	\begin{pmatrix}
			o(\boldsymbol{1}_2)\\
			o(\boldsymbol{1}_2)
		\end{pmatrix} .
	\end{equation*}
	Now, computing the inverse matrix
	\begin{equation*}
		\begin{pmatrix}
			\boldsymbol{J}_\beta(\boldsymbol{\theta}_0) & \boldsymbol{M}(\boldsymbol{\theta}_0) \\
			\boldsymbol{M}^T(\boldsymbol{\theta}_0) & \boldsymbol{0}
		\end{pmatrix}^{-1} = 
		\begin{pmatrix}
			\boldsymbol{P}_\beta(\boldsymbol{\theta}_0) & \boldsymbol{Q}_\beta(\boldsymbol{\theta}_0)\\
			\boldsymbol{Q}^T_\beta(\boldsymbol{\theta}_0) &  \boldsymbol{R}_\beta(\boldsymbol{\theta}_0)
		\end{pmatrix}
	\end{equation*} 
	with $\boldsymbol{P}_\beta(\boldsymbol{\theta}_0)$ and $\boldsymbol{Q}_\beta(\boldsymbol{\theta}_0)$ defined in (\ref{eq:SigmaPQmatrices}) and $\boldsymbol{R}_\beta(\boldsymbol{\theta}) = \left(\boldsymbol{M}^T(\boldsymbol{\theta})\boldsymbol{J}_\beta(\boldsymbol{\theta})^{-1}\boldsymbol{M}(\boldsymbol{\theta})\right)^{-1}$. But from the $\beta$-score function converges to a normal distribution of the form 
	\begin{equation*}
		\begin{pmatrix}
			\sqrt{N} \boldsymbol{U}_{\beta,N}(\boldsymbol{\theta}_0)\\
			\boldsymbol{0}_r
		\end{pmatrix} \xrightarrow[N\rightarrow \infty]{L} \mathcal{N}\left(\boldsymbol{0}_{3}, \begin{pmatrix}
			\boldsymbol{K}_\beta(\boldsymbol{\theta}_0) & \boldsymbol{0}\\
			\boldsymbol{0}^T & 0
		\end{pmatrix}\right)
	\end{equation*}
	and hence, 
	\begin{equation*}
		\begin{pmatrix}
			\sqrt{N}(\widetilde{\boldsymbol{\theta}}^\beta - \boldsymbol{\theta}_0)\\
			-\sqrt{N}\boldsymbol{\lambda}_\beta
		\end{pmatrix} \xrightarrow[N\rightarrow \infty]{L} \mathcal{N}\left(\boldsymbol{0}_{3},
		\boldsymbol{V}_\beta(\boldsymbol{\theta}_0) \right)
	\end{equation*}
	with 
	\begin{equation*}
		\boldsymbol{V}_\beta(\boldsymbol{\theta}_0) = \begin{pmatrix}
			\boldsymbol{P}_\beta(\boldsymbol{\theta}_0) & \boldsymbol{Q}_\beta(\boldsymbol{\theta}_0)\\
			\boldsymbol{Q}_\beta(\boldsymbol{\theta}_0)^T &  \boldsymbol{R}_\beta(\boldsymbol{\theta}_0)
		\end{pmatrix}
		\begin{pmatrix}
			\boldsymbol{K}_\beta(\boldsymbol{\theta}_0) & \boldsymbol{0}\\
			\boldsymbol{0}^T & 0
		\end{pmatrix}
		\begin{pmatrix}
			\boldsymbol{P}_\beta(\boldsymbol{\theta}_0) & \boldsymbol{Q}_\beta(\boldsymbol{\theta}_0)\\
			\boldsymbol{Q}_\beta(\boldsymbol{\theta}_0)^T &  \boldsymbol{R}_\beta(\boldsymbol{\theta}_0)
		\end{pmatrix}.
	\end{equation*}
	Thus, the asymptotic distribution of the vector of Lagrangian multipliers is given by
	\begin{equation*}\label{eq:asymplambda}
		\sqrt{N} \boldsymbol{\lambda}_\beta
		\xrightarrow[N\rightarrow \infty]{L} \mathcal{N}\left(0,
		\boldsymbol{Q}_\beta(\boldsymbol{\theta}_0)^T \boldsymbol{K}_\beta(\boldsymbol{\theta}_0)\boldsymbol{Q}_\beta(\boldsymbol{\theta}_0) \right).
	\end{equation*}
	Using the previous convergence and the consistency of the RMDPDE, it follows that the asymptotic distribution of the Rao-type test statistics, $\boldsymbol{R}_{\beta}(\widetilde{\boldsymbol{\theta}}^\beta),$
	is a chi-square distribution with $r$ degrees of freedom.
	
\end{proof}

If the null hypothesis defined in Equation (\ref{eq:null}) holds, the $\beta$-score function defined in Equation (\ref{eq:score}) should be small, and so Rao-type test statistic defined in (\ref{eq:raotest}). Then, we can define the reject region of the test as

$$RC = \left\{ (n_1,...,n_{L+1}): \boldsymbol{R}_{\beta}(\widetilde{\boldsymbol{\theta}}^\beta)> \chi^2_{r, \alpha}  \right\},$$
where $\chi^2_{r, \alpha}$ is the upper $\alpha$ quantile of chi-squared distribution with $r$ degrees of freedom.

The robustness of the estimators will be inherited by its associated test statistics, and so, large values of the tuning parameter will produce more robust test statistics.

Let us finally consider the most common composite hypothesis related to the problem of testing a part of the parameter vector,
\begin{equation*}
	\operatorname{H}_0 : \boldsymbol{\theta}^{1} = \boldsymbol{\theta}^{01},
\end{equation*}
where  $\boldsymbol{\theta}^{1}$ are $r < 3$ the components of $\boldsymbol{\theta},$ assumed to be the first $r$ components for simplicity.
For example, the above set-up can be used to  test whether the lifetime distribution is an exponential, by testing
$ \operatorname{H}_0 : \eta = 1.$
In that case, the matrix of derivatives $\boldsymbol{M}^T(\boldsymbol{\beta}) = [\boldsymbol{I}_{r\times r} \boldsymbol{0}_{r\times(3-r)}]$ and the Rao-type test statistic is simplified as
$$	R_{\beta}(\widetilde{\boldsymbol{\theta}}^\beta) = N  \boldsymbol{U}_{\beta}^1(\widetilde{\boldsymbol{\theta}}_\beta)^T
 \boldsymbol{K}^{-1}_{\beta, 11}(\widetilde{\boldsymbol{\theta}}_\beta) \boldsymbol{U}_{\beta}^1(\widetilde{\boldsymbol{\theta}}_\beta)$$
where $ \boldsymbol{K}_{\beta, 11}(\boldsymbol{\theta})$ denotes the $r\times r$ principal sub-matrix of $ \boldsymbol{K}_{\beta}(\boldsymbol{\theta}).$

\section{Reliability analysis for solar lighting devices \label{sec:simulatio}}

We empirically examine the performance of the proposed RMDPDEs and their corresponding Rao-type test statistics in a simulated dataset based on a simple step-stress ALT conducted in  \cite{Han2014} for analyzing the reliability of solar lighting devices.
In the solar lighting device dataset the stress factor was temperature of the testing chamber, which was noted to typically operate at $ 293K$ but increased  to $353K$ for the step-stress study. The temperature was standardized in the statistical analysis yielding the stress levels $x_1= 0.1$ and $x_2= 0.5.$
The original dataset consisted of  $n = 30$ prototypes and the total test duration was $\tau_2 =  20$ (in hundred hours).
Prior studies considered exponential or Weibull probability distributions with a constant shape parameter across different temperature levels  for the failure times, and assumed that the distribution scale parameter and the stress level are log-linearly related. The  estimated  parameters assuming exponential distribution were $ (\widehat{a}_0, \widehat{a}_1) = (3.6597, -2.4131)^T,$ or equivalently $ \boldsymbol{\theta} = (\widehat{a}_0, \widehat{a}_1, \widehat{\eta})^T = (3.6597, -2.4131, 1)^T.$  assuming Weibull distributions.

Our simulation set-up is based on the previous solar lighting devices dataset. We consider $N=200$ and $L=14$ inspection times (in hundred hours), including the termination time, namely $IT = 1, 3, 5, 7, 8, 9, 10, 12, 13, 14, 15, 17, 19$ and $IT=20.$ The temperature is increased from $x_1$ to $x_2$ at time $\tau_1 = 1,$ and  Weibull distributions with constant shape parameter regardless the temperature and scale parameter depending on the temperature with true values $\boldsymbol{\theta}_0 = (a_0, a_1, \eta) = (3.6597, -2.4131, 1.4)$  are assumed for the lifetime of the devices. 
The shape parameter is set to a value different than 1 (corresponding to the exponential distribution) so that the model belongs to a  proper Weibull distribution, as originally assumed in the experiment.

Following the discussion in Section \ref{sec:Weibull},  the recorded failure counts are generated from a multinomial distribution with probability vector defined as in Equation (\ref{eq:thprob}).
Besides, to evaluate the robustness of the estimators and test statistics, we contaminate the probability of failure  for a  $\varepsilon \%$ of the  devices under test.
For  multinomial distributions,
outlying observations may appear in the form of early or late failures of a group of devices. 
Then, for contaminating the multinomial data, the  probability of failure at one cell is increased (representing early failures) so the devices would fail more likely at that modified cell. The rest of the cell probabilities  are decreased proportionally so the probability vector sum up to 1.
 
Figure \ref{fig: unrestricted} (right) presents the mean squared error (MSE) of the estimated model parameters over $R=10.000$ repetitions of the experiment (weighting each component by it true value for matching the different magnitudes of the parameter), computed as follows:
 $$MSE_{\beta} = \parallel \frac{\widehat{\boldsymbol{\theta}}^\beta - \boldsymbol{\theta}_0 }{\boldsymbol{\theta}_0}\parallel _2.$$
 As shown, contamination can heavily affect the estimation of the parameters, while the MSE of the MDPDE with moderate and large values of the tuning parameter $\beta$ remains sufficiently low even under high contamination.
Additionally, in many situations we may be interested in estimating 
the MTTF instead of the distribution parameters. 
An estimation of the mean lifetime under normal operating temperature $x_1=0.1$ can be straightforward obtained by plugging-in the MDPDEs, $\widehat{\boldsymbol{\theta}}^\beta$ on the Weibull MTTF as follows 
$$\widehat{MTTF}^\beta = \exp(\widehat{a}_0^\beta+\widehat{a}_1^\beta x_0) \cdot \Gamma\left(1+\frac{1}{\widehat{\eta}^\beta}\right).$$
The MSE produced by the different MDPDE when estimating the $MTTF$ for increasing contamination rates are presented in Figure \ref{fig: unrestricted} (left). The robustness characteristics of the estimator are directly transferred to the estimation of the lifetimes characteristics and so the MTTF based on the MLE is quite far from the true MTTF of the product in the presence of contamination. Conversely, in the absence of contamination, the estimates of the MTTF with different values of $\beta$ do not differ significantly.
\begin{figure}[H]
	\centering
	\begin{subfigure}{0.45\textwidth}
		\includegraphics[scale = 0.4]{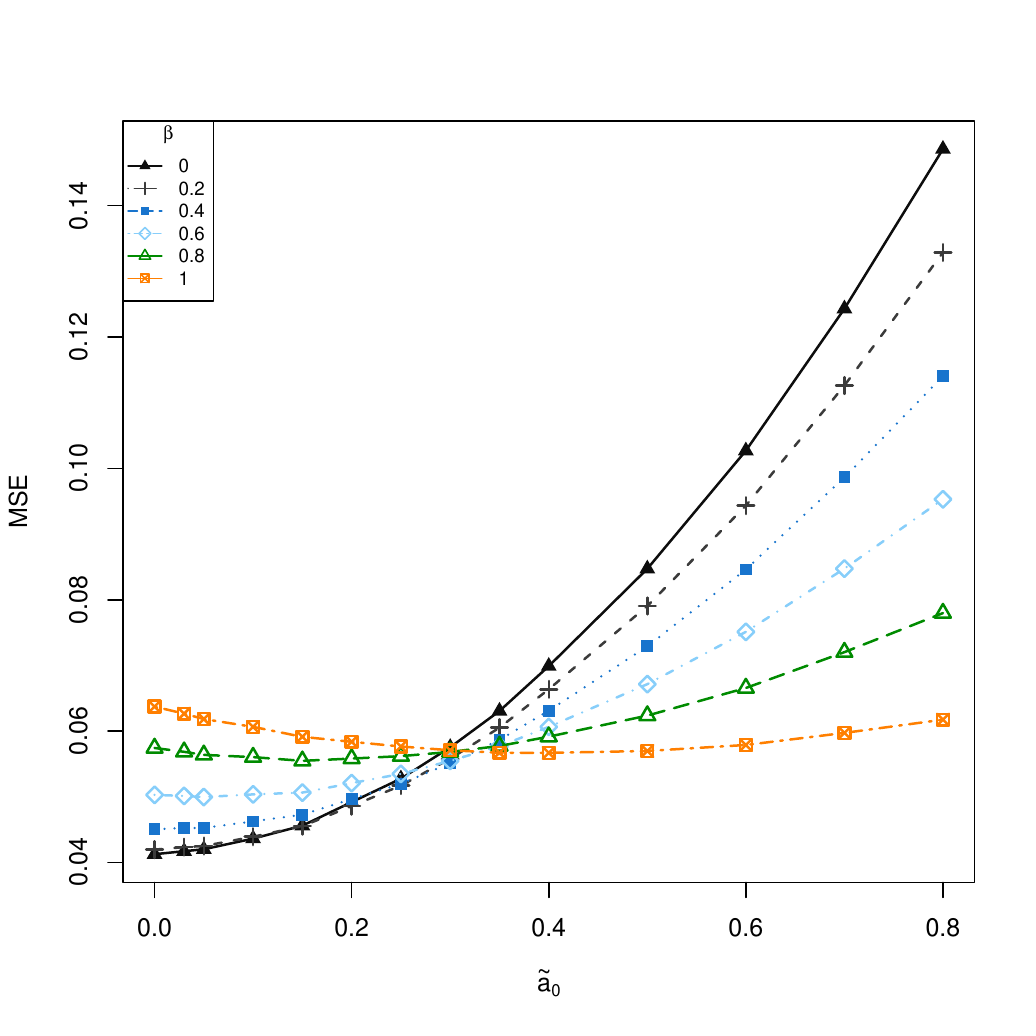}
		\caption{MSE of the model parameter estimates }
	\end{subfigure}
	\begin{subfigure}{0.45\textwidth}
		\includegraphics[scale = 0.4]{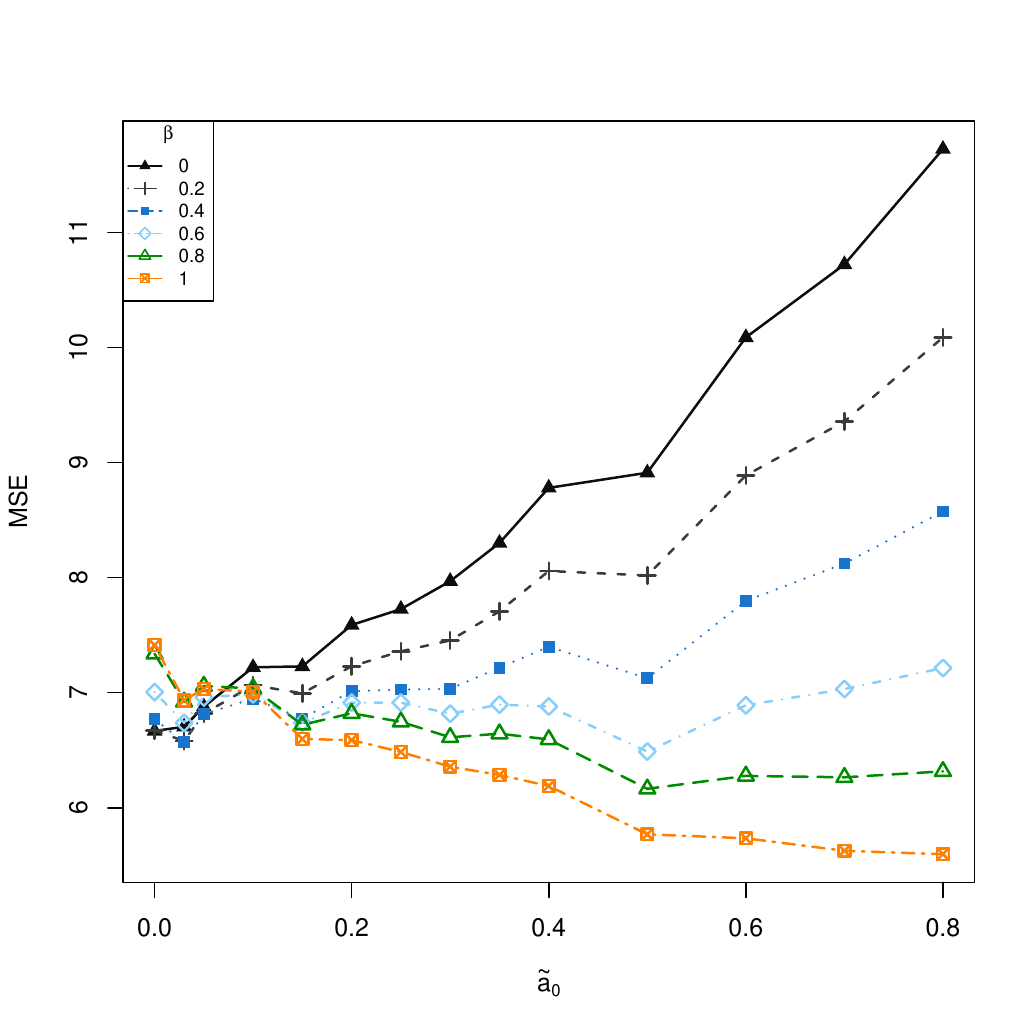}
		\caption{MSE of the estimated MTFF }
	\end{subfigure}
	\caption{Mean squared error (MSE) of the model estimates (left) and mean time to failure (right) for different values of the tuning parameter $\beta$ under increasing contamination rates }
	\label{fig: unrestricted}
\end{figure}

We now examine the behaviour of the restricted estimator under linear restrictions of the parameter space. In particular, we will define two different restricted spaces 
 of the form
\begin{equation*}
	\Theta_0 = \{\boldsymbol{\theta} : a_0 = 3.6597 \}, \hspace{0.3cm}\Theta_1 = \{\boldsymbol{\theta} : a_1 = -2.4131 \} \hspace{0.3cm} \text{and}\hspace{0.3cm} \Theta_2 = \{\boldsymbol{\theta} : \eta = 1.3 \}. 
\end{equation*}
These restrictions are defined by simple null hypothesis on each of the model parameters. Of course, as the true value of one of the three parameters is known, the estimation error would considerably decrease even under contamination. However, the robustness feature of the MDPDE is still remarkable and large values of the tuning parameter $\beta$ would produce more accurate estimation in the presence of contamination. Figure \ref{fig: restricted} plots the MSE of the estimates under increasing contamination rates. Here, the robustness of the estimators out stands under the  restricted space $\Theta_2$ (fixed shape parameter) which is a common assumption when assuming Weibull distribution for step-stress tests. As before, the lack of robustness is inherited to the estimation of any lifetime characteristic of interest, such as the MTTF and a similar performance than in the unrestricted estimation of the MTTF is obtained for the restricted case. We do not present the graphics of the MSE for brevity.

\begin{figure}[H]
	\centering
	\begin{subfigure}{0.3\textwidth}
		\includegraphics[scale = 0.29]{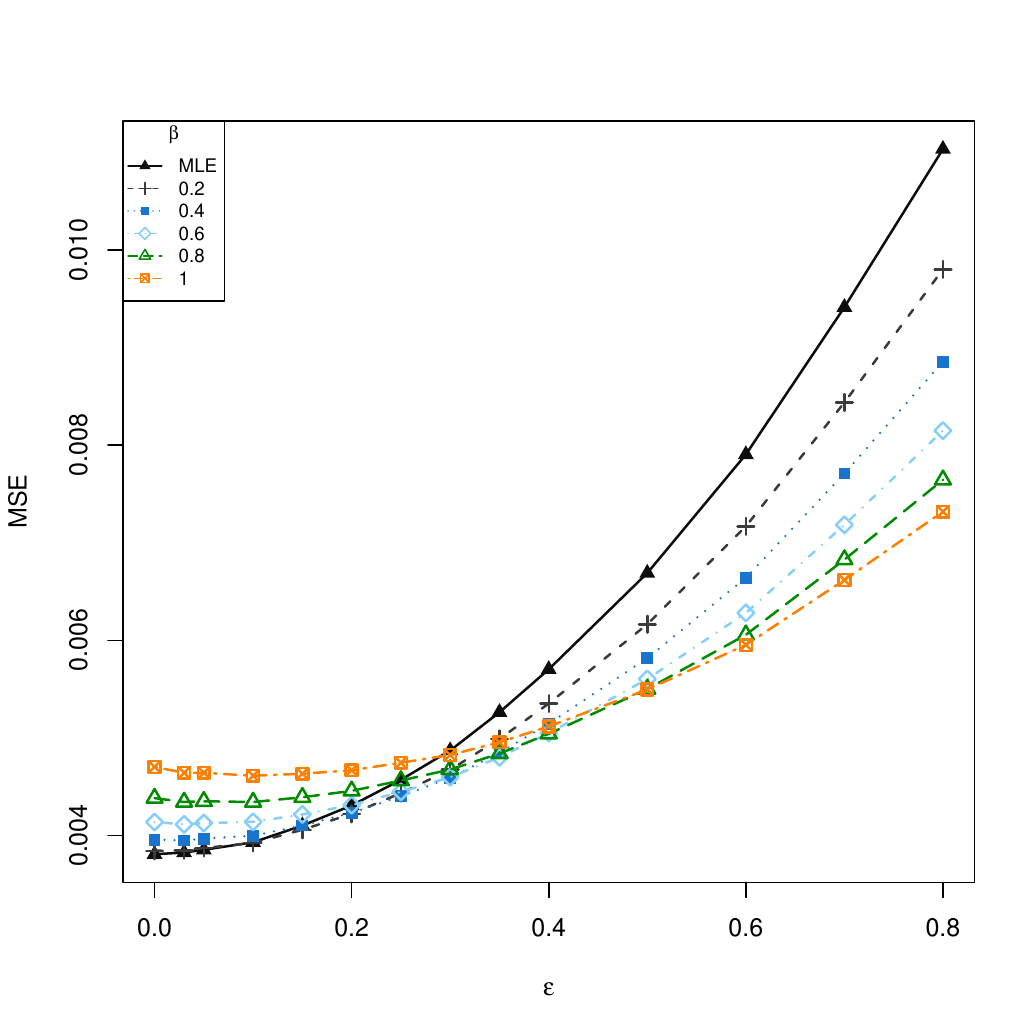}
		\caption{MSE of the RMDPDE under $\Theta_0$ }
	\end{subfigure}
	\begin{subfigure}{0.3\textwidth}
		\includegraphics[scale = 0.29]{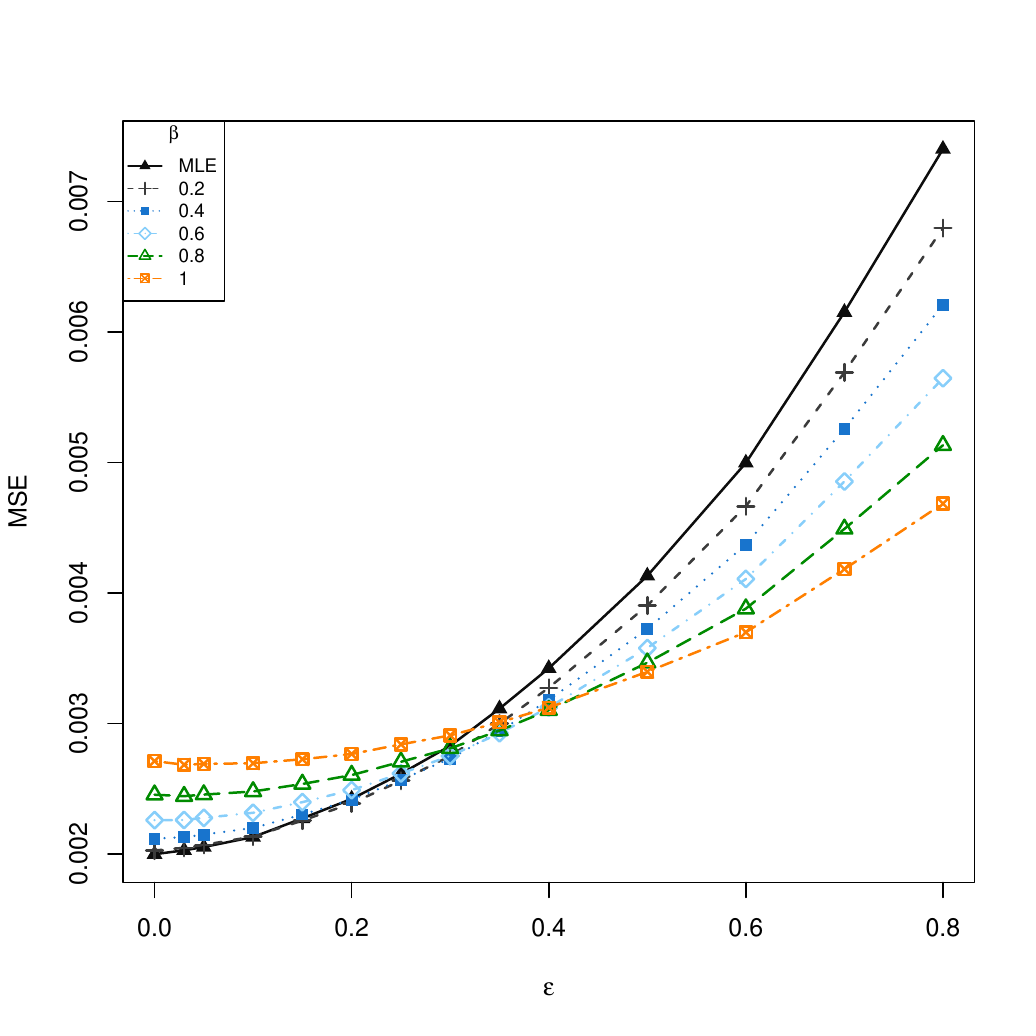}
		\caption{MSE of the RMDPDE under $\Theta_1$ }
	\end{subfigure}
\begin{subfigure}{0.3\textwidth}
	\includegraphics[scale = 0.29]{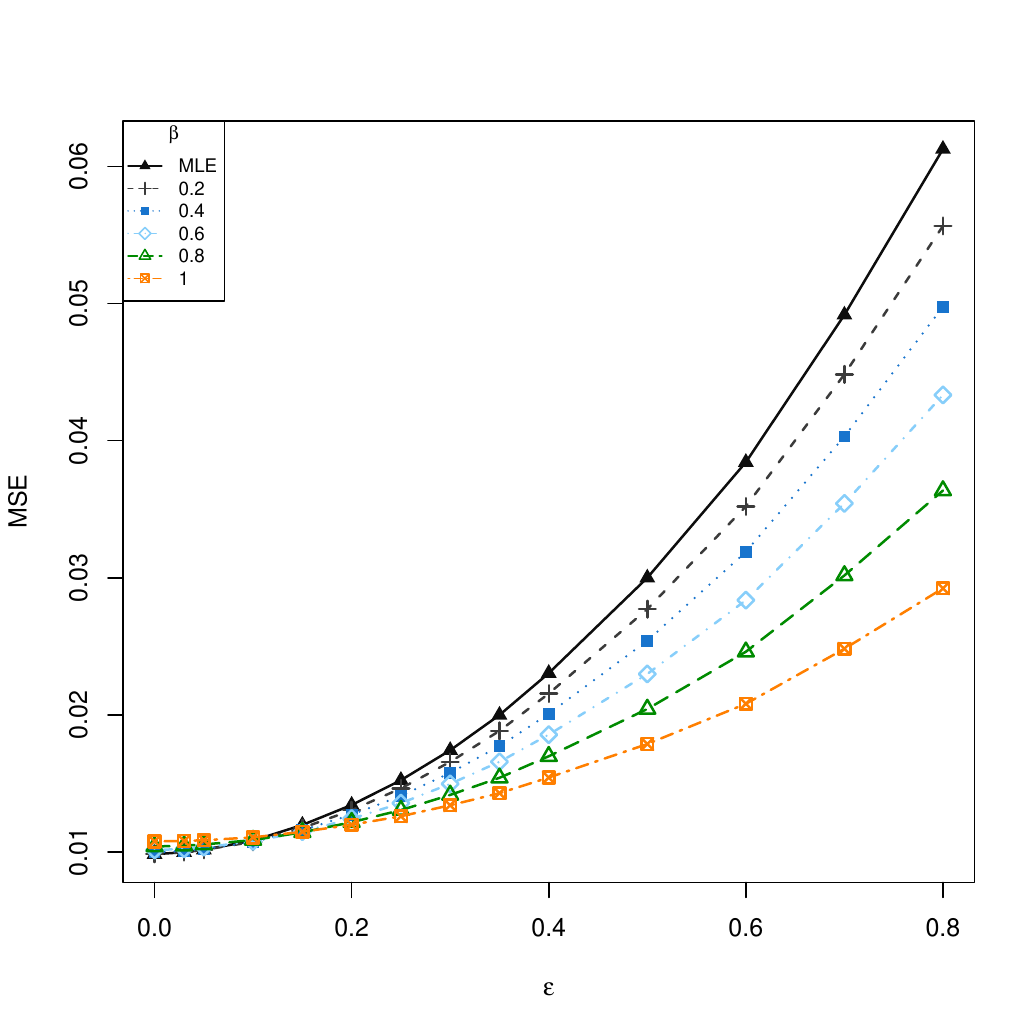}
	\caption{MSE of the RMDPDE under $\Theta_2$ }
\end{subfigure}
	\caption{Mean squared error (MSE) of the RMDPDE for different values of the tuning parameter $\beta$ under increasing contamination rates}
	\label{fig: restricted}
\end{figure}

Further, let us consider the testing problems with null hypothesis associated with the three restricted spaces are naturally defined as
\begin{equation*}
	(\operatorname{P}_0) \operatorname{H}_0  : a_0 = 3.6597 , \hspace{0.3cm} (\operatorname{P}_1)  \operatorname{H}_0  : a_1 = -2.4131 \hspace{0.3cm} \text{and} \hspace{0.3cm} (\operatorname{P}_2)  \operatorname{H}_0  : \eta = 1.3 . 
\end{equation*}
To evaluate the robustness of the Rao-type test statistics based on the RMDPDEs, Figure \ref{fig: raolevel} the empirical level of the Rao-type test statistics for the three testing problems $\operatorname{P}_0, \operatorname{P}_1$ and $\operatorname{P}_2$ with a significance level of $\alpha=0.05.$
\begin{figure}[H]
	\centering
	\begin{subfigure}{0.3\textwidth}
		\includegraphics[scale = 0.29]{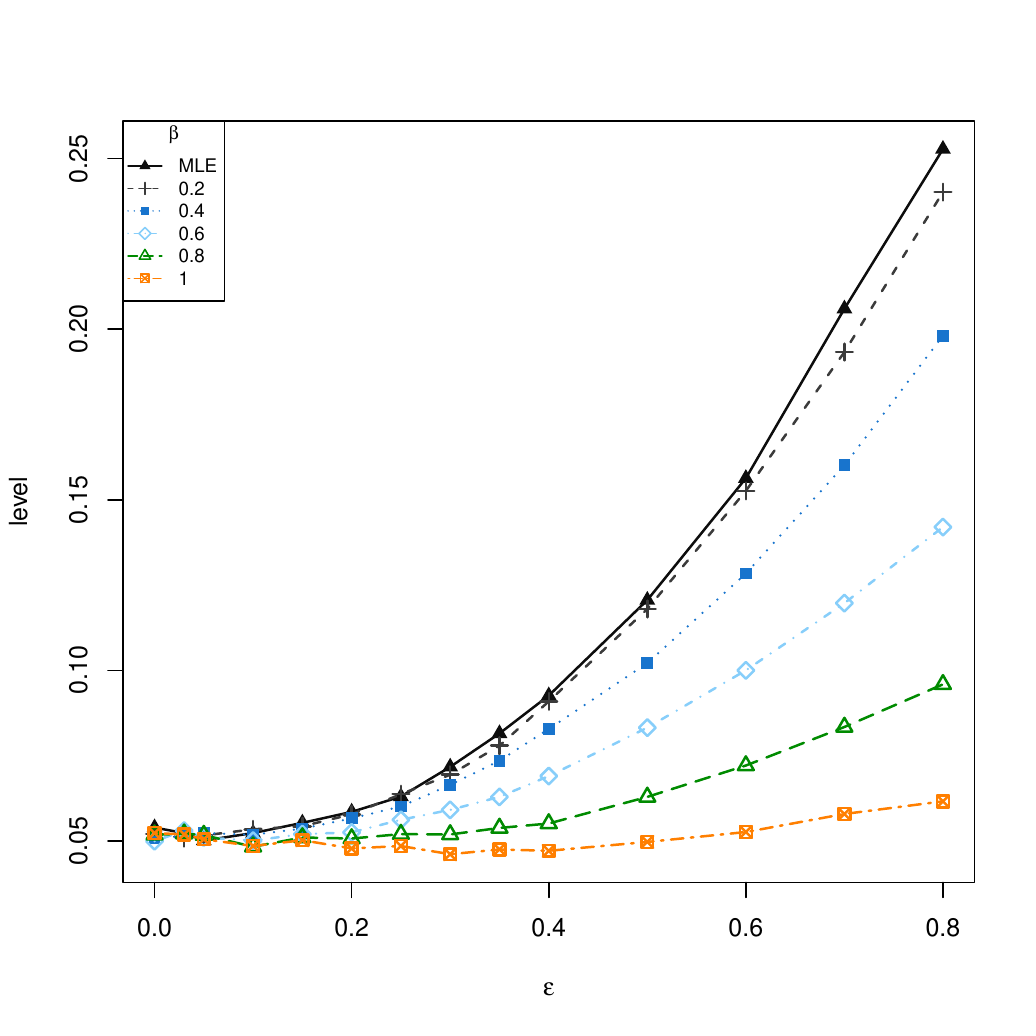}
		\caption{Empirical level of the Rao-type test statistic based on the RMDPDE for testing $\operatorname{P}_0$ }
	\end{subfigure}
	\begin{subfigure}{0.3\textwidth}
		\includegraphics[scale = 0.29]{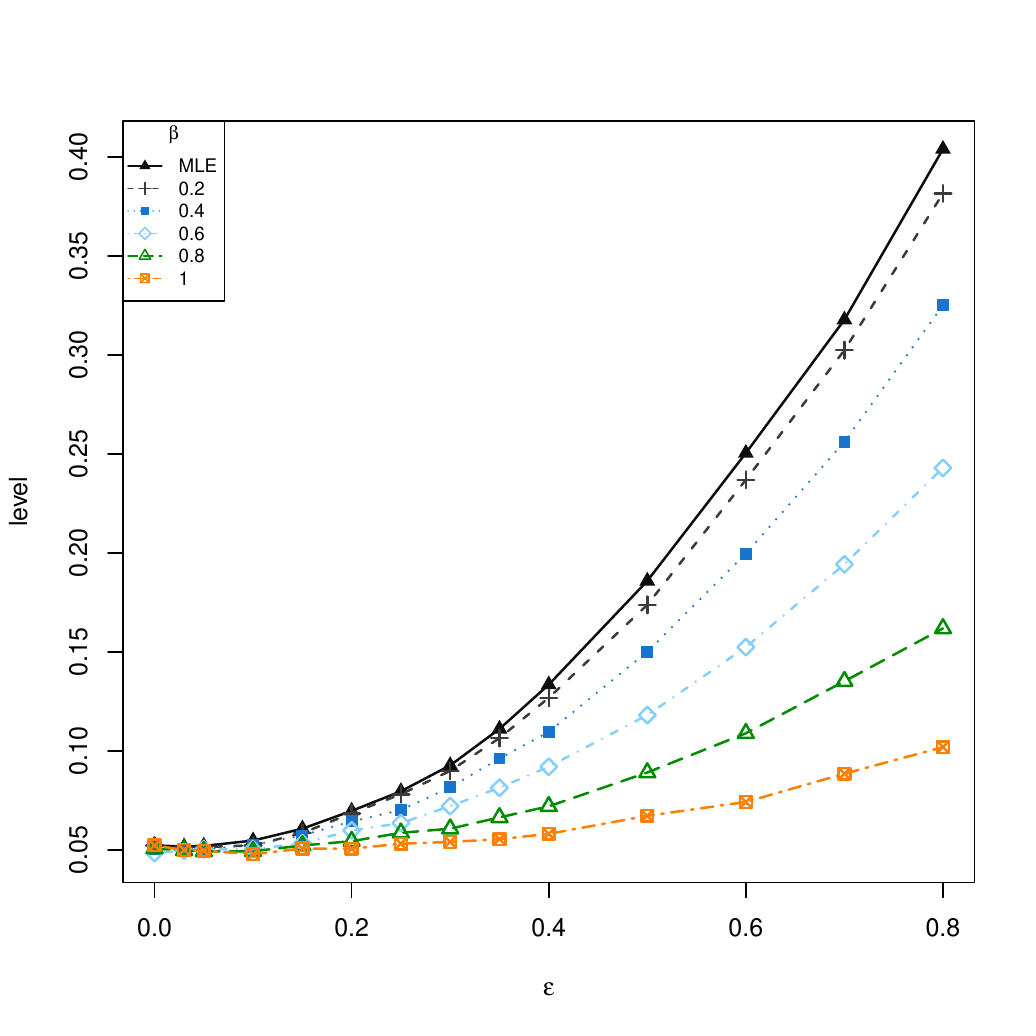}
		\caption{Empirical level of the Rao-type test statistic based on the RMDPDE for testing $\operatorname{P}_1$ }
	\end{subfigure}
	\begin{subfigure}{0.3\textwidth}
		\includegraphics[scale = 0.29]{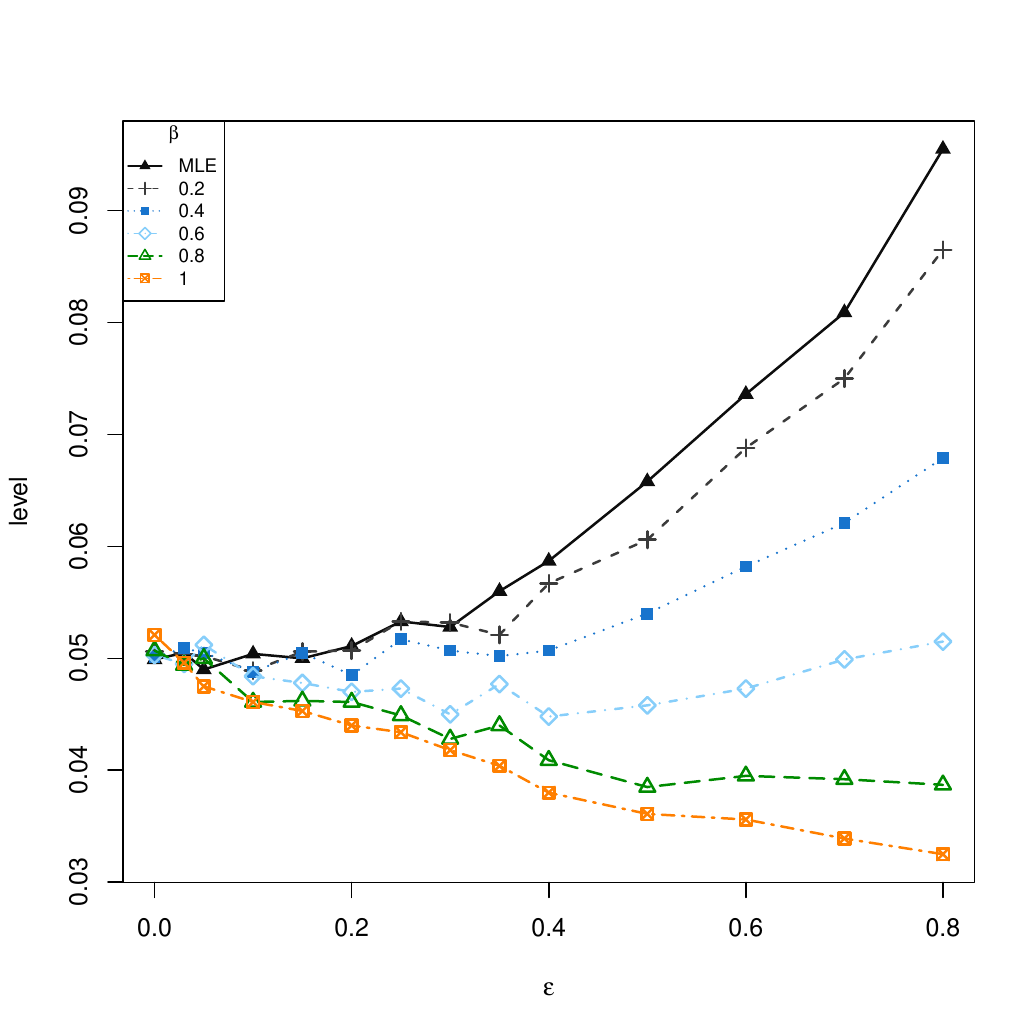}
		\caption{Empirical level of the Rao-type test statistic based on the RMDPDE for testing $\operatorname{P}_2$ }
	\end{subfigure}
	\caption{Empirical level of the Rao-type test statistics based on the RMDPDEs for different values of the tuning parameter $\beta$ under increasing contamination rates}
	\label{fig: raolevel}
\end{figure}
The lack of robustness of the MLE and RMDPDE with low values of $\beta$ is intensified in their associated Rao-type test statistics, which achieves empirical levels of $\alpha = 0.3$ under heavily contaminated scenarios. In contrast, Rao-type test statistics with moderate and large values of $\beta$ keep acceptable in those scenarios.
Of course, the larger the tuning parameter $\beta$ is, the more robust but less efficient the resulting estimator and test statistic will be. From our results, moderate values of $\beta$, over 0.4 could result in a great compromise between efficiency and robustness. This  suggestion for the $\beta$ value is in line with the DPD literature.

Finally, a natural hypothesis arising when assuming Weibull lifetimes is to test whether the assumption about the lifetime distribution can be simplified to an exponential distribution, as done in many analysis of the solar light data. In that case, the testing problem is of the form
\begin{equation*}
	 (\operatorname{P}_2^\ast)  \operatorname{H}_0  : \eta = 1. 
\end{equation*}
and the null hypothesis does not hold. Then,  the behaviour of the Rao-type test statistic for testing $(\operatorname{P}_2^\ast)$ illustrates the performance of the power function of the test. Figure \ref{fig: raopower} shows the empirical power of the  Rao-type test statistics based on the RMDPDEs for different values of the tuning parameter under increasing contamination rates. Under pure data, the power of the tests based on the RMLE ($\beta=0$) is higher than any other power based on the RMDPDEs with positive values of $\beta$, although both clearly reject the null hypothesis with a high p-value. However, when introducing contamination the performance of the Rao-test statistic based on the RMLE worsens considerably. However, the Rao-type test statistics based on robust RMDPDEs with moderate and large values of the tuning parameter keeps competitive even under heavily contaminated scenarios, showing the robustness advantage of the proposed methods.
\begin{figure}[H]
	\centering
	\includegraphics[scale = 0.29]{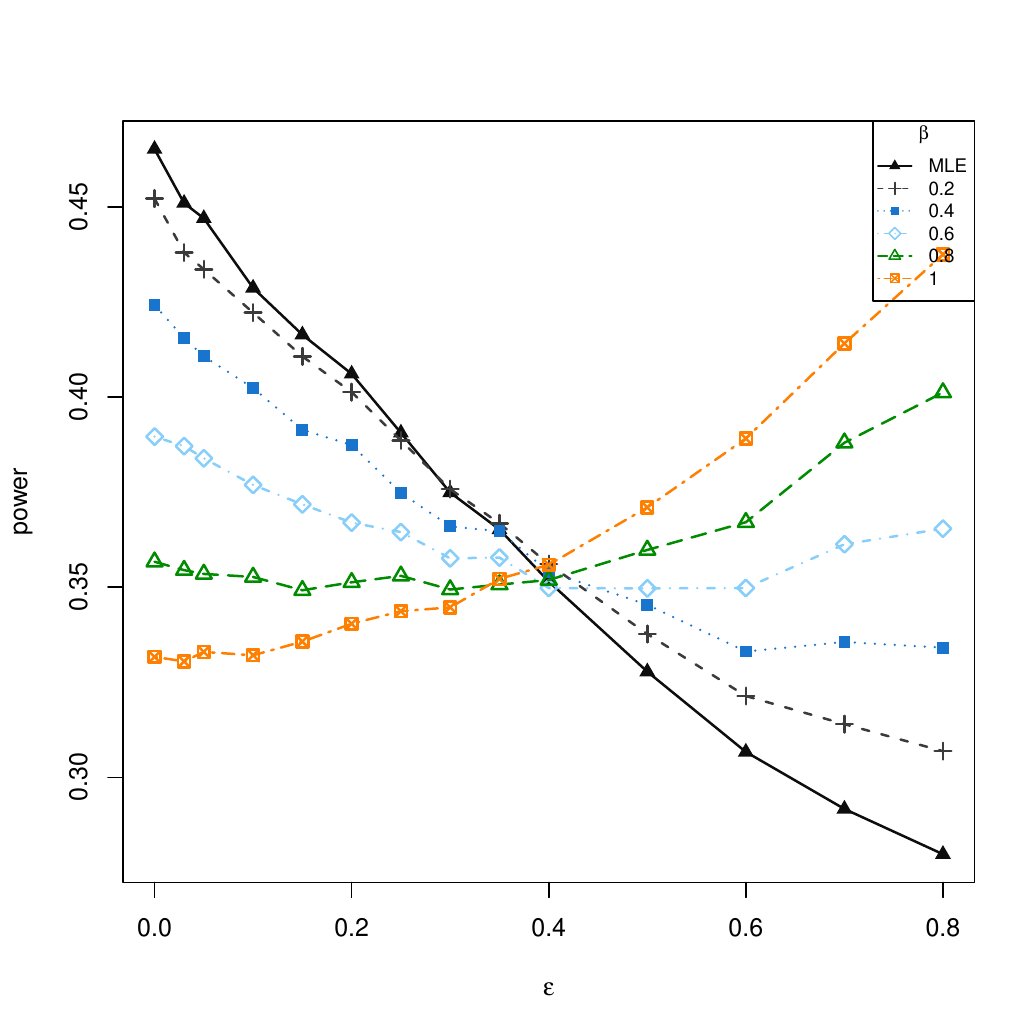}
	\caption{Empirical power of the Rao-type test statistics based on the RMDPDEs for different values of the tuning parameter $\beta$ under increasing contamination rates}
	\label{fig: raopower}
\end{figure}





\section{Concluding remarks}

Restricted estimators are necessary in some statistical applications such as the definition of Rao-type test statistics, which are widely used in hypothesis testing.
Inferential techniques based on maximum likelihood and minimum DPD have been recently developed for interval-censored step-stress ALT models.
The experimental observed data can be modelled according to a multinomial model, where the multinomial events are defined by the different intervals inspected. In those models, data contamination and outlying failures counts can heavily influence the model parameter estimates and so robust estimators are suitable for inference.
In this paper we have developed robust RMDPDEs for interval-censored step-stress ALTs under Weibull lifetime distributions under restricted parameter spaces with general equality constraints. The RMDPDEs are proved theoretically and empirically to be consistent and robust against data contamination.
Moreover, based on the robust RMDPDEs we have presented a robust generalization of the Rao test statistic for testing general composite null hypothesis.
The proposed family of restricted estimators includes the restricted MLE as a special case and so, the classical Rao test for the step-stress ALT model under interval-censorship is also examined.
Finally, the advantage of the robust restricted estimators and  their applicability for hypothesis testing in practice have been illustrated with a simulation set-up based on a real experiment conducted for analyzing the reliability of solar lighting devices.

\bibliography{bibliography}

\end{document}